\documentclass[preprint,aps,prd,floatfix]{revtex4-2}
\usepackage{graphicx}
\usepackage{amsfonts}
\usepackage{amsmath}
\usepackage{rotating}
\usepackage{amssymb,amsthm}
\usepackage{xcolor}
\usepackage[normalem]{ulem}

\newtheorem{thm}{Theorem}[section]
\newtheorem{cor}[thm]{Corollary}
\newtheorem{lem}[thm]{Lemma}

\theoremstyle{definition}
\newtheorem{defn}[thm]{Definition}
\theoremstyle{remark}
\newtheorem{rem}[thm]{Remark}

\newcommand{\mbold}[1]{\mbox{\boldmath{\ensuremath{#1}}}}

\def\beq{\begin{eqnarray}}
\def\eeq{\end{eqnarray}}

\def \bell {\mbox{{\mbold\ell}}}
\def \bn {\mbox{{\bf n}}}
\def \bm {\mbox{{\bf m}}}
\def \bh {\mbox{{\bf h}}}
\def \bomega {\mbox{{\mbold \omega}}}

\def \bV {\mbox{{ {\bf V}}}}
\def \bA {\mbox{{ {\bf A}}}}

\begin{document}


\title{Teleparallel geometries not characterized by their scalar polynomial torsion invariants}

\author {D. D. McNutt}
\email{david.d.mcnutt@uis.no}
\affiliation{ Faculty of Science and Technology, University of Stavanger, N-4036 Stavanger, Norway }

\author {A. A. Coley}
\email{aac@mathstat.dal.ca}
\affiliation{Department of Mathematics and Statistics, Dalhousie University, Halifax, Nova Scotia, Canada, B3H 3J5}

\author{R. J. \surname{van den Hoogen}}
\email{rvandenh@stfx.ca}
\affiliation{Department of Mathematics and Statistics, St. Francis Xavier University, Antigonish, Nova Scotia, Canada, B2G 2W5}



\begin{abstract}

A teleparallel geometry is an n-dimensional manifold equipped with a frame basis and an independent spin connection. For such a geometry, the curvature tensor vanishes and the torsion tensor is non-zero. A straightforward approach to characterizing teleparallel geometries is to compute scalar polynomial invariants constructed from the torsion tensor and its covariant derivatives. An open question has been whether the set of all scalar polynomial torsion invariants, $\mathcal{I}_T$ uniquely characterize a given teleparallel geometry. In this paper we show that the answer is no and construct the most general class of teleparallel geometries in four dimensions which cannot be characterized by $\mathcal{I}_T$. As a corollary we determine all teleparallel geometries which have vanishing scalar polynomial torsion invariants. 

\end{abstract}

\maketitle

\section{Introduction}

There continues to be interest in alternative theories
to General Relativity (GR) \cite{Nojiri_Odintsov2006,Capozziello_DeLaurentis_2011}. While a conservative extension to GR arises in theories based on Riemann-Cartan geometries, where the torsion tensor is permitted to be non-zero in addition to the curvature tensor, there are other possibilities. For example one could introduce a more general structure, such as in Finsler geometries where there are now several curvature and torsion tensors \cite{stavrinos1999some, Ikeda:2019ckp} or to consider geometries where the curvature tensor vanishes and either the torsion tensor or the non-metricity tensor are non-zero \cite{Conroy:2017yln}. One particular class of gravitational theories assumes that the dynamics of the gravitational field are no longer encoded within the metric and its corresponding Levi-Civita connection but are encoded through a coframe basis, and its corresponding field strengths represented by  the torsion  \cite{Li_Sotiriou_Barrow2010, Krssak2015, Bahamonde_Boehmer_Wright2015}. These gravitational models are commonly labelled as teleparallel theories of gravity. There exists a subclass of teleparallel theories of gravity that are dynamically equivalent to GR called the teleparallel equivalent to GR (TEGR),  in which the Lagrangian is based on a scalar, $T$, constructed from the torsion and differs from the Lagrangian of GR by a total derivative (implying the the field equations are formally equivalent). A common extension of TEGR, so-called $f(T)$ gravity, is based on a generalized Lagrangian where instead of having just a scalar $T$ in the Lagrangian, it is replaced with an arbitrary function of the scalar $T$.

We will consider covariant formulations of $f(T)$ gravity in a gauge invariant manner by assuming a spin connection that has zero curvature and which is not trivial. The most general spin-connection that satisfies this requirement is the {\it purely inertial spin-connection} \cite{Obukhov_Pereira2003,Obukhov_Rubilar2006,Lucas_Obukhov_Pereira2009,Aldrovandi_Pereira2013,Krssak:2018ywd}. With this understanding, the spin-connection only vanishes in a very special class of frames (``proper frames'') where all inertial effects are absent \cite{Obukhov_Pereira2003,Obukhov_Rubilar2006,Lucas_Obukhov_Pereira2009,Aldrovandi_Pereira2013,Krssak:2018ywd}. The main advantage of the {\it covariant} teleparallel gravity approach is that by using the purely inertial connection, the resulting teleparallel gravity theory embodied by the Lorentz covariant field equations, is locally Lorentz invariant \cite{Krssak_Saridakis2015,Krssak_Pereira2015}.

In teleparallel gravity, the coframe basis  together with a flat or inertial spin connection essentially replaces the metric and its corresponding Levi-Civita connection as the primary objects of study.
 Therefore, to proceed in building gravitational models and finding solutions, one requires a choice of coordinates, $x^\mu$,
a coframe basis, $\bh^a$,  and a flat or inertial  spin-connection, (possibly utilizing a proper frame basis in which the components of the spin-connection are trivial). Given the variety of choices to be made in the selection of coordinates, frames, and connections, different choices may yet yield equivalent solutions. The question therefore arises as to how can one determine whether two solutions, which appear different and employ different coordinates, co-frames and connections, are inequivalent.  Fortunately two solutions can be shown to be inequivalent by comparing the group of symmetries for each solution as was described in \cite{Coley:2019zld}. 

Even if two teleparallel geometries have the same group of symmetries  this does not imply that they are equivalent. To show that two teleparallel geometries are equivalent we must compute and compare invariant quantities associated with each teleparallel geometry. In \cite{Coley:2019zld}, we proposed a modification of the Cartan-Karlhede  equivalence algorithm (which is applicable to {\it any} teleparallel geometry, regardless of the field equations) to determine the equivalence of two teleparallel geometries. 

As an alternative to the calculation of symmetries or usage of the Cartan-Karlhede algorithm to partially determine equivalence, we can instead consider any scalar invariant produced from contractions of copies of the torsion tensor and its covariant derivatives (hereafter called torsion tensors) known as a {\it scalar polynomial torsion invariant} (TSPI). TSPIs have the useful property that they are invariant under Lorentz transformations and hence are frame invariant scalars. 
The set of all such TSPIs, $\mathcal{I}_T$, can be used to uniquely characterize (locally) most teleparallel geometries.  Those teleparallel geometries which fail to be uniquely characterized by the set of all TSPIs we will call $\mathcal{I}_T$-degenerate.

A teleparallel geomtery is $\mathcal{I}_T$-degenerate if there exists a deformation of the proper coframe $\bh^a_\tau$ such that $\bh^a_\tau$ is continuous in $\tau$, $\bh^a_0 = \bh^a$ and the limiting proper coframe as $\tau \to \infty $, $\bh^a_\infty $ is not diffeomorphic to $\bh^a$ but the set of TSPIs, $\mathcal{I}_T$, for $\bh^a$ and $\bh^a_\infty $ are identical. Thus, the proper coframe $\bh^a$ cannot be distinguished from $\bh^a_\tau$ using TSPIs. A more practical definition of $\mathcal{I}_T$-degeneracy can be stated in terms of the alignment classification of the torsion tensors using non-proper frames, and this will be discussed in section \ref{sec:Alignment}. 

In the following we will determine the entire class of $\mathcal{I}_T$-degenerate teleparallel geometries. The outline of the paper is as follows. Section \ref{sec:GeoFrame} gives a review of teleparallel geometries. In section \ref{sec:CKalg}, we outline the construction of Cartan invariants and TSPIs. In section \ref{sec:Alignment}  we give a concise review of the alignment classification and its relevance for describing $\mathcal{I}_T$-degenerate spacetimes, and we also define three useful choices of coframe for a given teleparallel geometry. We provide two explicit examples of $\mathcal{I}_T$-degenerate teleparallel geometries in section \ref{sec:Examples}. In section \ref{sec:Ideg}, we show that any $\mathcal{I}_T$-degenerate teleparallel geometry must belong to a class of teleparallel geometries that is an analogue of the degenerate Kundt class in GR with a restricted form for the spin-connection. In addition we fully describe the class of VSI$_T$ teleparallel geometries where all TSPIs vanish. In section \ref{sec:discussion} we review our results and discuss their implications.

\section{The Geometrical Framework} \label{sec:GeoFrame}
 
We use the following notation. Coordinate indices, which range from 1 to 4, are represented by $\mu, \nu, \ldots$ while the tangent space indices, which also range from 1 to 4, are labelled by $a,b,\ldots$.  Relative to the frame basis, covariant differentiation with respect to the connection is indicated by a vertical bar $T_{abc|e}$ or as $D_{e} T_{abc}$. We will also write $D_{\mu } = h^a_{~\mu} D_a$ and $\partial_{a } = h_{a}^{~\mu} \partial_\mu$.  Covariant differentiation with the coordinate basis will be written as $\nabla_{\mu} T_{\nu \lambda \gamma}$. The connection relative to the coframe basis, $$D_c \bh^a = -\omega^a_{~bc} \bh^b,$$ is related to the coordinate basis connection through the identity: 
\beq \omega^a_{~bc} = - h_b^{~\nu} h_c^{~\mu} ( \partial_{\mu} h^a_{~\nu} - \Gamma^{\lambda}_{~\nu \mu} h^a_{~\lambda}) = -h_b^{~\nu} h_c^{~\mu} \nabla_{\mu} h^a_{~\nu}. \eeq
Covariant differentiation with respect to the Levi-Civita connection, $\tilde{D}_a$, or $\tilde{\nabla}_\mu$ is indicated with a semi-colon: $T_{abc;e}$ and $\tilde{\nabla}_{a} = h_a^{~\mu} \tilde{\nabla}_\mu$.  Round and square brackets denote symmetrization and anti-symmetrization, respectively, where underlined indices are not included in any symmetrization.

Let $M$ be a 4D differentiable manifold with coordinates $x^\mu$. A basis for the tangent space at each point in the manifold can be expressed as $\bh_a=h_a^{~\mu}\,\partial_\mu $ which we call the frame, while the corresponding dual space of one-forms are $\bh^a=h^a_{~\mu}\,dx^\mu$ which we will call the coframe. The associated vielbein expressions $h_a^{~\mu}$ and $h^a_{~\mu}$ are non-singular matrices of functions of $x^\mu$ and provide a mechanism to express the components of tensors in terms of the coordinate basis $\{\partial_\mu,dx^\nu\}$ or in terms of the tangent space basis $\{\bh_a,\bh^b\}$. We shall complete the geometrical setting by assuming the existence of a metric and a spin connection one-form
\begin{eqnarray}
\bf{g} &=& g_{\mu\nu}\,dx^\mu \otimes dx^\nu = g_{ab}\,\bh^a \otimes \bh^b\\
\bomega^a_{~b} &=& \omega^a_{~b\nu} \,dx^\nu = \omega^a_{~bc} \,\bh^c.
\end{eqnarray}

In general, the geometrical quantities ${\bf g},\bh^a,\bomega^{a}_{\phantom{a}b}$ are independent.  However, assuming the Principle of Relativity implies that the geometry is invariant under local $GL(4,\mathbb{R})$ gauge transformations of the coframe (or frame) allows for some efficiencies.  We can use this gauge freedom to express the components of the metric with respect to the tangent space basis in a {\em{preferred}} form. Two useful possibilities, are the
\begin{eqnarray} 
\text{\bf Orthonormal gauge} \qquad && g_{ab} \to \eta_{ab} =  \mbox{Diag}[-1,1,1,1] \qquad \text{and\ the} \\
\text{\bf Complex Null gauge} \qquad && g_{ab} \to \eta_{ab} = \left[ \begin{array}{cccc} 0 & -1 & 0 & 0 \\ -1 & 0 & 0 & 0 \\ 0 & 0 & 0 & 1 \\ 0 & 0 & 1 & 0 \end{array}\right].
\end{eqnarray}

Taking an orthonormal coframe, $\{\bh^a\}$, the complex null coframe, $\{ \bn, \bell,  \bar{\bm}, \bm\}$ is constructed as
\beq \bell = \frac{1}{\sqrt{2}} (\bh^1 - \bh^2), \bn = \frac{1}{\sqrt{2}} (\bh^1 - \bh^2), \bm =\frac{1}{\sqrt{2}} (\bh^3 - i \bh^4). \label{NullCoframe} \eeq 
\noindent We shall denote the dual complex null frame as $\{ {\bf L}, {\bf N}, {\bf M}, \bar{{\bf M}}\}$. In either case we label the Minkowski metric as $\eta_{ab}$ while the corresponding coordinates of the metric are $g_{\mu\nu}=\eta_{ab}h^{a}_{~\mu}h^b_{~\nu}$, illustrating the use of the vielbein functions.  

The metric components, $\eta_{ab}$, while fixed, are still invariant under the residual $SO^{+}(1,3)$ subgroup of $GL(4,\mathbb{R})$ gauge transformations of the co-frame.  These local Lorentz transformations of the coframe  \begin{equation}
\bh^a \to \Lambda^a_{~b} \bh^b.
\end{equation} 
which leave the metric unchanged, are represented by the matrix $\Lambda^a_{~b}(x^\mu)$,  with a corresponding inverse of $\left(\Lambda^{-1}\right)^a_{\ b}=\Lambda_b^{\ a}$. Therefore, upon a choice of gauge, any field equations and any tensorial expressions must transform homogeneously under these Lorentz transformations. Note, all connections, even those with a frame where all components of the spin-connection are zero, transform non-homogeneously under linear transformations as
\begin{equation}
\omega^a_{~bc} \to\left(\Lambda^a_{~d}\omega^{d}_{~ef}\Lambda_b^{~e}+\Lambda^a_{~d}\partial_f\Lambda_b^{~d}\right)\Lambda_{c}^{~f}.
\end{equation}

With the additional assumption that the spin-connection be metric compatible, i.e., $D_{c} g_{ab} =0$, the spin-connection one-form becomes anti-symmetric, $\bomega_{(ab)}=0$
(which we assume hereafter).  The torsion and the curvature tensors associated with the coframe and spin-connection are:
\begin{eqnarray}
T^a_{\phantom{a}\mu\nu}&=&\partial_\mu h^a_{\phantom{a}\nu}-\partial_\nu h^a_{\phantom{a}\mu}+\omega^a_{\phantom{a}b\mu}h^b_{\phantom{a}\nu}-\omega^a_{\phantom{a}b\nu}h^b_{\phantom{a}\mu},\\
R^a_{\phantom{a}b\mu\nu} &=& \partial_\mu \omega^a_{\phantom{a}b\nu}-\partial_\nu \omega^a_{\phantom{a}b\mu}+\omega^a_{\phantom{a}c\mu}\omega^c_{\phantom{a}b\nu}-\omega^a_{\phantom{a}c\nu}\omega^c_{\phantom{a}b\mu}.\label{curvature}
\end{eqnarray}
Since we are interested in teleparallel geometries, the curvature tensor $R^a_{\phantom{a}b\mu\nu}$, is assumed to be identically zero.  In which case, equation \eqref{curvature} can be solved for the spin connection
\begin{equation}
\omega^a_{~bc}=\Lambda^a_{~d}\partial_c(\Lambda_b^{~d})
\end{equation}
for some matrix $\Lambda^a_{~b}(x^\mu) \in SO^{+}(1,3)$. This zero curvature or flat connection is also called an inertial connection \cite{Krssak:2018ywd}.

For a given teleparallel geometry, $(M, \bh^a, \bomega^a_{~b})$ it is always possible to apply a Lorentz transformation to set $\bomega^a_{~b} = {\bf 0}$ \cite{Obukhov_Pereira2003,Obukhov_Rubilar2006,Lucas_Obukhov_Pereira2009,Aldrovandi_Pereira2013,Krssak:2018ywd}. This motivates the definition of a special class of frames, known as {\bf proper frames}, where the components of the spin connection vanish,  so that $\nabla_\mu h^c_{~\nu} = 0$ which is equivalent to $D_a \bh^b = {\bf 0}$. It is always possible to represent any exact solution to teleparallel gravity in a proper frame. Relative to this frame, only global Lorentz transformations will preserve the proper frame condition.

\section{Torsion Invariants} \label{sec:CKalg}

\subsubsection{The modified Cartan-Karlhede algorithm for torsion} 

A modification of the Cartan-Karlhede algorithm was introduced for teleparallel geometries in \cite{Coley:2019zld}. In this section we will briefly review the approach and define terminology that will be used in the remainder of the paper. We will denote $\mathcal{T}^q$ as the set of components of the torsion tensor and the covariant derivatives of the torsion tensor up to the $q$-$th$ covariant derivative, 
\beq \mathcal{T}^q = \{ T_{abc}, T_{abc|d_1}, \ldots T_{abc|d_1 \ldots d_q} \}. \eeq   For any teleparallel geometry the Cartan-Karlhede algorithm may be summarized by the following steps:

\begin{enumerate}
\item Set the order of differentiation $q$ to 0.
\item Calculate $\mathcal{T}^q$.
\item Determine the canonical form of the $q$-th covariant derivative of the torsion tensor.
\item Fix the frame as much as possible, using this canonical form, and record the remaining frame transformations that preserve this canonical form (the group of allowed frame transformations is the {\it linear isotropy group $H_q$}). The dimension of $H_q$ is the dimension of the remaining {\it vertical} freedom of the frame bundle.
\item Find the number $t_q$ of independent functions of spacetime position in $\mathcal{T}^q$ in the canonical form. This tells us the remaining {\it horizontal} freedom.
\item If the dimension of $H_q$ and number of independent functions are the same as in the previous step, let $p+1=q$, and the algorithm terminates; if they differ (or if $q=0$), increase $q$ by 1 and go to step 2. 
\end{enumerate}

\noindent The resulting non-zero components of $\mathcal{T}^p$ constitute the {\it Cartan invariants} and we will denote them as $\mathcal{T} \equiv \mathcal{T}^{p+1}$ so that 
\beq \mathcal{T} = \{ T_{abc}, T_{abc|d_1}, \ldots T_{abc|d_1 \ldots d_{p+1}} \}. \label{CartanSet} \eeq 
We will refer to the invariants constructed from, or equal to, Cartan invariants of any order as {\it extended Cartan invariants}. 

This algorithm will pick out a preferred frame for a given teleparallel geometry which is adapted to the form of the torsion tensor and its covariant derivatives. This motivates the following definitions:

\begin{defn}
Any frame determined in a coordinate-independent manner using the Lorentz frame transformations to fix the torsion tensor and its covariant derivatives (hereafter, we will refer to these tensors as torsion tensors) into canonical forms is called an {\bf invariantly defined Cartan frame up to linear isotropy $H_p$} which we shall refer to as an {\bf Cartan frame} in brief here. 
\end{defn}

In general, there is non-trivial linear isotropy and so the frame may not be completely determined. We note that the existence of non-trivial linear isotropy in the CK algorithm is necessary and sufficient for a non-trivial isotropy subgroup in the group of affine frame symmetries \cite{Coley:2019zld}.

If the linear isotropy group is trivial, $H_{p} = \emptyset$ then the Cartan frame is fully determined by the teleparallel geometry. 

\begin{defn} 
If a Cartan frame can be constructed by fixing all of the parameters of the  Lorentz frame transformations, then such a frame is called an {\bf invariant frame}: 
\end{defn}

\noindent We note that the linear isotropy is inherited by all geometric objects, including the spin connection.
If there is some remaining Lorentz gauge freedom, it is only possible to determine an invariant frame through prolongation of the manifold \cite{olver1995}, which is not ideal for the study of gravitational theories. 

\begin{rem}
We remark that a {\it proper frame} is not necessarily a {\it Cartan frame} as its definition relies on the vanishing of the spin-connection which is not Lorentz covariant. It is possible to choose the proper frame and then completely fix the remaining Lorentz parameters.
\end{rem}

For sufficiently smooth frames and spin-connection, the result of the algorithm is a set of Cartan scalars providing a unique local geometric characterization of the teleparallel geometry. The 4D spacetime is characterized by the canonical form used for the torsion tensors, the two discrete sequences arising from the successive linear isotropy groups and the independent function counts, and the values of the (non-zero) Cartan invariants. As there are $t_p$ essential spacetime coordinates, the remaining $n-t_p$ are ignorable, and so the dimension of the affine frame symmetry isotropy group (hereafter called the isotropy group) of the spacetime will be $s=\dim(H_p)$ and the affine frame symmetry group has dimension: \beq r=s+n-t_p. \label{rnumber} \eeq

\subsection{Scalar torsion polynomial invariants}

Just as in the case of Riemannian geometries based on the metric alone, there exists a set of scalar invariants that can be constructed from the torsion tensor and its covariant derivatives. These invariants, called {\it  scalar polynomial torsion invariants} (TSPIs), are constructed from full contractions of tensors built from the torsion tensor and its covariant derivatives with respect to the spin connection. We will denote the set of all such TSPIs as $\mathcal{I}_T$. While this set is functionally infinite there is a finite dimensional basis for $\mathcal{I}_T$.  We define two special cases: Vanishing Scalar Invariant teleparallel geometries, ($VSI_T$) and Constant Scalar Invariant teleparallel geometries, ($CSI_T$) as a geometry in which all of the TSPIs are zero and constant, respectively.

A natural question to ask is whether the TSPIs are able to uniquely characterize teleparallel geometries, or if there are  cases where the TSPIs are unable to distinguish between different teleparallel geometries. For example, it is possible that there are teleparallel geometries where all TSPIs vanish, and hence have the VSI$_T$ property \cite{{Higher}}. For such VSI$_T$ teleparallel geometries, it would  then be impossible to distinguish them from Minkowski space using TSPIs. A related question is to ask when are the Cartan invariants uniquely determined by TSPIs \cite{Coley:2009eb} and when is the frame bundle completely characterized by TSPIs.

An important example of an TSPI is the torsion scalar:
\begin{equation}
T =
\frac{1}{4} \; T^a_{\phantom{a}bc} \, T_a^{\phantom{a}bc} +
\frac{1}{2} \; T^a_{\phantom{a}bc} \, T^{c b}_{\phantom{aa}a} -
               T^a_{\phantom{a}c a} \, T^{b c}_{\phantom{aa}b},
\label{TeleLagra}
\end{equation}
that is employed in TEGR and $f(T)$ theories of gravity.  Generalizations of TEGR and $f(T)$ theories of gravity, can be constructed using other TSPIs:
\beq \mathcal{L}_{Grav} (h^a_{~\mu}, \omega^a_{~b \mu} ) = \frac{h}{2	\kappa} f( T, T_1, T_2, \ldots T_n), T_i \in \mathcal{I}_T,~ i \in [1,n]. \eeq
By including a matter Lagrangian and varying the resulting action, we can generate new teleparallel gravity theories that are distinct from the $f(T)$ theories. We emphasize that the geometric or kinematic structures on which a physical theory is constructed will, in general, be independent of the physical and dynamical features of the theory. This is illustrated by the fact that for a given teleparallel geometry we may impose the action of TEGR in addition to the action of other $f(T)$ theories.

\section{The alignment classification and TSPIs} \label{sec:Alignment}

In order to determine the class of teleparallel geometries which cannot be classified by their TSPIs, we must examine a frame based approach to the classification of tensors in Lorentzian manifolds. Relative to a given null complex coframe $\{  \bn, \bell, \bar{\bm}, \bm \}$, the abelian subgroup of the group $SO^{+}(1, 3)$ consists of a boost,
\beq
(\bell, \bn) \to (e^{\lambda} \bell, e^{-\lambda} \bn), \label{badboost}
\eeq
\noindent where $\lambda$ is real-valued.
This gives rise to the concept of a  {\it boost weight} $b \in \mathbb{Z}$ such that for an arbitrary component of a rank $r$ tensor ${\bf \hat{T}}$ with respect to the null coframe \eqref{NullCoframe}, a boost in each of the $n$ planes, $(\bell, \bn)$, gives the transformation \cite{OrtaggioPravdaPravdova:2013}:
\beq \hat{T}_{a_1 ... a_r} \to e^{b\lambda} \hat{T}_{a_1 ... a_r},  \eeq
\noindent where the indices $a_i, \ldots a_r$ range from $1$ to $4$, and the integer $b$ is the boost weight vector of the component $\hat{T}_{a_1 ... a_r}$ which records the difference in the number of appearances of $\bell$ and $\bn$, respectively, in the associated tensor product of a given component $\hat{T}_{a_1 ... a_r}$. We can write the tensor ${\bf \hat{T}}$ in the following decomposition:
\beq
{\bf \hat{T}} = \sum\limits_{b} ({\bf \hat{T}})_{b}.
\eeq
\noindent Here $({\bf \hat{T}})_{b}$ denotes the projection onto the subspace of components of boost weight $b$.

Defining the maximum boost weight of a tensor, ${\bf \hat{T}}$, for a null direction $\bell$ as the boost order, we denote this as $\mathcal{B}_{{\bf \hat{T}}}(\bell)$. For a given null direction $\bell$, $\mathcal{B}_{{\bf \hat{T}}} (\bell)$ remains unchanged under boosts, spatial rotations and null rotations about $\bell$, implying that the choice of $\bn$ will not affect the integer value of $\mathcal{B}_{{\bf \hat{T}}} (\bell)$. This implies the definition is dependent on the choice of $\bell$. Defining $B_{{\bf \hat{T}}}$ as the maximum value of $\mathcal{B}_{{\bf \hat{T}}} (\bell)$ over all possible choices of $\bell$, the existence of a $\bell$ with $\mathcal{B}_{{\bf \hat{T}}} (\bell) <B_{{\bf \hat{T}}}$ is an invariant property of the tensor ${\bf \hat{T}}$. We will say $\bell$ is ${\bf \hat{T}}$-aligned if $\mathcal{B}_{{\bf \hat{T}}}(\bell) < B_{{\bf \hat{T}}}$.

The torsion tensor and any rank two tensor, ${\bf T}$, can be broadly classified into five {\it alignment types}: if for all null directions $\bell$, $\mathcal{B}_{{\bf T}} (\bell)=2$ then ${\bf T}$ is of alignment type $G$, if there exists an $\bell$ such that $\mathcal{B}_{{\bf T}} (\bell) = 1, 0,-1,-2$ then ${\bf T}$ is of alignment type $I, II, III,$ or $N$, respectively, while if {\bf T} vanishes then it belongs to alignment type $O$. We will define the {\it secondary alignment type} as the alignment classification of the remaining null direction $\bn$. For higher rank tensors, like the covariant derivatives of the torsion tensor, the alignment types are still applicable despite the possibility that $|\mathcal{B}_{{\bf T}} (\bell)|$ may be greater than two; in particular, we will be interested in tensors which have alignment type {\bf II} or more special.  

Using the boost weight decomposition, we can introduce properties to classify tensors in a similar manner to the alignment classification for general pseudo-Riemannian spaces \cite{OrtaggioPravdaPravdova:2013,HHY}. 

\begin{thm} \label{thm:HHY1}
	A tensor of arbitrary rank is not characterized by its invariants if and only if it is of alignment type {\bf II} or more special.
\end{thm}

\noindent Using this theorem we can characterize $\mathcal{I}_T$-degenerate teleparallel geometries using the boost weight decomposition.

\begin{cor} \label{cor:Ideg}
	A teleparallel geometry is $\mathcal{I}_T$-degenerate if and only if the torsion tensors are of alignment type {\bf II} or more special relative to some common null coframe.
\end{cor}

In analogy with the results of $\mathcal{I}_{\tilde{R}}$-degenerate Lorentzian spacetimes \cite{CHPP2009}, we will suppose that there exists a null complex coframe, $\{ \bn, \bell,  \bar{\bm}, \bm\}$ with $\bell$ aligned with the torsion tensor and the first covariant derivative of $\bell$ such that both are of alignment type {\bf II}, and  we assume
\begin{equation}
\bell_{a|b} = L \bell_{a}\bell_{b} +L_i(\bell_{a} m^i_{b} + m^i_{a}\bell_{b})\label{ell}
\end{equation}
\noindent where $L$ and $L_i$ ($i=1,2$) are arbitrary functions. Then there always exists a null coordinates $u$ and $v$ such that $\bh^2 = \bell  \equiv du$ and the corresponding frame member is $\bh_1 = {\bf L} \equiv \frac{\partial}{\partial v}$. Relative to the coordinates $\{v,u,x,y\}$, the null coframe takes the form:
\beq \bh^a = \left[ \begin{array}{c} \bn \\ \bell \\ \bar{\bm} \\ \bm \end{array} \right] = \left[ \begin{array}{c}   dv + H du + W_1 dx + W_2 dy \\ du \\ {P}(dx-idy)  \\ {P}(dx+idy) \end{array} \right], \label{AKundtframe}\eeq
\noindent where $H(u,v,x,y)$, $W_1(u,v,x,y)$ and $W_2(u,v,x,y)$ are arbitrary functions and $P=P(u,x,y)$. 

We will define an {\bf aligned Kundt coframe} as a null coframe \eqref{AKundtframe} together with a non-trivial spin connection such that $\bell$ is automatically aligned with the torsion tensor and the covariant derivative of $\bell$. A teleparallel geometry with this choice of frame may not be $\mathcal{I}_T$-degenerate as the covariant derivative of the torsion tensor could be of alignment type {\bf I}. While this choice of coframe puts strong conditions on the spin connection: \beq \omega^1_{~11} = \omega^3_{~31}, \bomega^2_{~3} = {\bf 0}, \bomega^2_{~4} = {\bf 0}, \bomega^3_{~1} = {\bf 0}, \bomega^4_{~1} = {\bf 0}, \label{SCzeroAlignedKundt}\eeq or equivalently in terms of the Newman-Penrose (NP) scalars, $\kappa = \rho = \sigma = \tau = \epsilon = 0$, this will not specify the remaining part of the spin connection. To fully describe the teleparallel geometry we must impose additional conditions on the spin connection. 

As an alternative, we will begin from the assumption of a proper null coframe for an $\mathcal{I}_T$-degenerate teleparallel geometry where the torsion tensors may not be explicitly of alignment type {\bf II}. As this geometry is $\mathcal{I}_T$-degenerate, from corollary \ref{cor:Ideg} there exists a null direction $\bell'$ for which the torsion tensors are all of alignment type {\bf II} relative to a coframe adapted to this direction. We can always apply a null rotation about $\bn$ to align the torsion tensors with this particular null direction. As a helpful definition, we will call any frame related to a proper frame by a null rotation about $\bn$, a {\bf $\bell$-proper frame}. Relative to the $\bell$-proper frame, the corresponding non-zero spin connection one-form components are: 
\beq \bomega^2_{~3}, \bomega^2_{~4} \label{lproperspin} \eeq
\noindent and their antisymmetric counterparts. Equivalently the NP scalars $\kappa, \rho, \sigma$ and $\tau$ are non-trivial.

Using the $\bell$-{\bf proper frame}, we can determine the subclass of teleparallel geometries in which there exists a frame, called a {\bf degenerate Kundt frame}, where all of the positive b.w. components of the torsion tensor and spin connection (and hence the components of all the covariant derivatives of the torsion tensor) are simultaneously zero. That is, the torsion tensors are all of type {\bf II} in the same frame.

\subsection{Irreducible parts of the torsion tensor}

The torsion two-form, $T^a$, can be expanded as
\begin{equation}
T^a=\frac{1}{2}T^a_{\phantom{a}bc} \,h^b \wedge h^c
\end{equation}
where $T^a_{~bc}$ has $24$ independent components. Under the local Lorentz group, the torsion tensor can be decomposed into three irreducible parts \cite{Hehl_McCrea_Mielke_Neeman1995}:
\beq T_{abc} = \frac23 (t_{abc} - t_{acb}) - \frac13 (g_{ab} V_c - g_{ac} V_b) + \epsilon_{abcd} A^d.\label{TorsionDecomp} \eeq

\noindent Here ${\bf V}$ denotes the vector part which is the trace of the torsion tensor, which has been written as a one-form:
\beq V_a = T^b_{~ba}, \label{Vtor} \eeq
\noindent Lowering the index of the torsion tensor and applying the Hodge dual of the resulting tensor gives the axial part, ${\bf A}$:
\beq A^a = \frac16 \epsilon^{abcd}T_{bcd}. \label{Ator} \eeq
\noindent Finally, we can construct the purely tensor part, ${\bf t}$:
\beq t_{(ab)c} = \frac12 (T_{abc}+ T_{bac}) -\frac16 (g_{ca} V_b + g_{cb} V_a) + \frac13 g_{ab} V_c. \label{Ttor} \eeq
\noindent These tensors are known as the {\it vector part, axial part, and tensor part } of the torsion tensor. The tensor part satisfies the following identities:
\beq \begin{aligned}
& g^{ab} t_{(ab)c} = 0, ~t_{(ab)c} = t_{(ba)c},~t_{(ab)c} + t_{(bc)a} + t_{(ca)b} = 0. \end{aligned} \label{tensortorsionids} \eeq

\noindent A simple counting argument shows that the tensor part of the torsion will have 16 algebraically independent components. Including the components of the vector part and axial part, this gives 24 components for the torsion tensor.

Working with a complex null coframe, $\{ \bn, \bell, \bar{\bm}, \bm \}$, \noindent the following basis of algebraically independent components can be chosen for the purely tensorial part of torsion:
\beq \begin{aligned} & t_{(12)1},~ t_{(13)1},~ t_{(14)1},~ t_{(22)1},~ t_{(23)1},~ t_{(24)1},~ t_{(33)1},~ t_{(44)1},~ \\
& t_{(13)2},~ t_{(14)2},~ t_{(23)2},~ t_{(24)2},~ t_{(33)2},~ t_{(44)2},~ t_{(14)3},~ t_{(24)3},  \end{aligned}\eeq

\noindent the remaining components of $t_{(ab)c}$ have algebraic dependencies arising from the trace-free and cyclic identities in equation \eqref{tensortorsionids} and share the same boost weight as the algebraically independent components. The boost weight of the algebraically independent components are given in figure 1

\begin{figure} \label{fig:bwd1torsion}
\beq \begin{array}{c|c} b.w. & Components \\ \hline
2 & t_{(13)1}, t_{(14)1} \\
1 &  t_{(12)1}, t_{(33)1}, t_{(44)1}, t_{(14)3} \\
0 & t_{(23)1}, t_{(24)1}, t_{(13)2}, t_{(14)2} \\
-1 & t_{(22)1}, t_{(33)2}, t_{(44)2}, t_{(24)3} \\
-2 & t_{(23)2}, t_{(24)2}  \end{array} \eeq
\caption{The b.w. values of the algebraically independent components of the purely tensor part of torsion in an arbitrary frame.}
\end{figure}

We will first exploit the existence of the vector and axial parts of the torsion, $\bV$ and $\bA$. If ${\bf V}$ and ${\bf A}$ are non-trivial and if the torsion tensor is of alignment type {\bf II} then noting that $\epsilon_{abcd}$ is of alignment type {\bf II}, then the vector and axial part of the torsion are necessarily of alignment type {\bf II} or {\bf III}. That is, these vectors can be spatial or null, but not timelike. The contributions of the vector and axial parts to the torsion tensor are

\beq \hat{T}_{abc} = \frac23 g_{a[b}V_{c]} + \epsilon_{abcd}A^d \label{VAtensor} \eeq

Unlike the vector parts of the torsion tensor, the tensor part of the torsion is not easily classified; however, it is possible to apply the alignment classification. By subtracting the vector and axial parts using equation \eqref{VAtensor} from the torsion $T_{abc}$, it follows that the purely tensorial part of the torsion must be at most alignment type {\bf II}. 

\begin{lem} \label{lem:AlgSpcl}
A torsion tensor is of alignment type {\bf II} or more special if and only if relative to a common null frame:
\begin{itemize} \item The purely tensor-part of the torsion tensor, $t_{(ab)c}$, is of alignment type {\bf II}, {\bf III}, {\bf N} or {\bf O}.
\item The vector and axial parts of torsion are of alignment type {\bf II}, {\bf III} or {\bf O}.
\end{itemize}
\end{lem}

\section{An example: vanishing scalar invariants } \label{sec:Examples}

As a motivating example, we will consider a class of teleparallel geometries belonging to one of the simplest classes of $\mathcal{I}_T$-degenerate telelparallel geometries where all TSPIs vanish, the VSI$_T$ teleparallel geometries. In order to determine some examples of VSI$_T$ teleparallel geometries, we will consider the teleparallel analogue of the VSI$_{\tilde{R}}$ spacetimes in GR where all scalar polynomial  curvature invariants ($\tilde{R}$SPIs) vanish. We note that in this case we are considering all $SPIs$ which are constructed from the curvature tensor, $\tilde{R}_{abcd}$ and its covariant derivatives with respect to the Levi-Civita connection. 

The teleparallel analogues of the VSI$_{\tilde{R}}$ spacetimes can be described using the aligned Kundt coframe \cite{Higher}:
\beq \bh^a = \left[ \begin{array}{c} \bn \\ \bell \\ {\bf \bar{m}} \\ {\bf m} \end{array} \right] = \left[ \begin{array}{c}   dv + H du + W_1 dx + W_2 dy \\ du \\ \frac{1}{\sqrt{2}}(dx - i dy) \\ \frac{1}{\sqrt{2}}(dx + i dy) \end{array} \right], \label{alignedKundtCoframe} \eeq
\noindent where $H(u,v,x,y)$, $W_1(u,v,x,y)$ and $W_2(u,v,x,y)$ are of the form ( and $P=1$)
\beq \begin{aligned} & H = \frac{\epsilon^2}{2x^2} v^2 + H^1(u,x,y) v + H^0(u,x,y) \\
& W_1 = -\frac{2 \epsilon }{x} v + W_1^0(u,x,y),~W_2 =  W_2^0(u,x,y). \end{aligned} \eeq
\noindent Here, $\epsilon = 0$ or $1$, and the remaining metric functions $H^1, H^0, W_1^0$ and $W_2^0$ are arbitrary. 

We note that this class of teleparallel geometries contains the plane parallel gravitational wave (PPGW) teleparallel geometries as a special case \cite{Coley:2019zld}. The PPGW teleparallel geometries are a frame-based analogue of the vacuum PP-wave spacetimes which admit an isotropy group of null rotations about $\bell$, known as the plane gravitational waves metrics in GR \cite{kramer}.

We have not yet specified a spin connection which will contribute to the form of the torsion tensors. To initiate the analysis and consequent observations we will consider an inertial spin connection  in which the the above frame becomes a proper frame.

The torsion is of the form:
\beq T^a = \left[ T^1, 0, 0, 0 \right]^T \eeq
\noindent where
\beq \begin{aligned} T^1 =& \left( H^ 1 + \frac{\epsilon^2 v}{x^2} \right)  ( \bh^1 \wedge \bh^2) - \left( \frac{\sqrt{2} \epsilon}{x} \right)  ( \bh^1 \wedge \bh^3) - \left( \frac{\sqrt{2} \epsilon}{x} \right)  ( \bh^1 \wedge \bh^4) \\  & - \left( \frac{\sqrt{2} \mathcal{V} }{x} \right)  ( \bh^2 \wedge \bh^3) - \left( \frac{\sqrt{2} \bar{\mathcal{V}} }{x} \right)  ( \bh^2 \wedge \bh^4) \\ & - i \left( \frac{ 2W_2^0 \epsilon -W^0_{2,x} x + W^0_{1,y} x }{x} \right)  ( \bh^3 \wedge \bh^4), \end{aligned} \eeq

\noindent with $\mathcal{V} $ a polynomial in $v$,
\beq \begin{aligned} \sqrt{2} x^3 \mathcal{V} =&  -[H^1_{~,x} x^3 - W^0_1 \epsilon^2 x + i (W^0_2 \epsilon^2 x - H^1_{~,y} x^3)]v \\
& + W^0_1 H^1 x^3 +W^0_{~1,u}x^3 +H^0_{,x} x^3 - 2 H^0 \epsilon x^2 \\ & - i(W_2^0 H^1 x^3 + W^0_{~2,u} x^3 - H^0_{~,y} x^3 ). \end{aligned} \eeq

\noindent We note when $\epsilon \neq 0$, not all TSPIs will vanish since for example the trace-part ${\bf V}$ has the form

\beq {\bf V} = \left(H^1 + \frac{\epsilon^2 v}{x^2} \right) \bell - \frac{\epsilon \sqrt{2}}{x} \bm - \frac{\epsilon \sqrt{2}}{x} \bar{\bm}, \label{VSIv} \eeq
\noindent and so the magnitude of ${\bf V}$, $|{\bf V}|^2 = 2 \epsilon |x|^{-1}$, is a non-vanishing and non-constant TSPI. Therefore this teleparallel geometry is VSI$_T$ only when $\epsilon = 0$.

This is a departure from what is found in GR.  Using co-frame \eqref{alignedKundtCoframe} with its corresponding Levi-Civita spin connection, one finds that all Levi-Civita scalar curvature invariants are zero and therefore this geometry is a VSI$_{\tilde R}$.  But as we see above, VSI$_{\tilde R}$ spacetimes typically have non-trivial TSPIs in general.  The teleparallel analogue of the VSI$_{\tilde R}$ geometries can have non-zero members in ${\cal I}_T$, for the zero spin connection, and hence are not VSI$_T$. It is possible that a judicious choice of the spin connection could correct the $\epsilon \not= 0$ case above and ensure that all TSPIs vanish as well. We will return to this question after determining the class of all ${\cal I}_T$-degenerate teleparallel geometries.

\section{$\mathcal{I}_T$-degenerate teleparallel geometries} \label{sec:Ideg}

We will consider the class of $\mathcal{I}_T$-degenerate teleparallel geometries in a non-proper frame. From  corollary \ref{cor:Ideg}, a teleparallel geometry will be $\mathcal{I}_T$-degenerate if the torsion tensor and its covariant derivatives are of alignment type {\bf II} to all orders. This will restrict the form of the frame and chosen spin connection.

To continue we will utilize the NP formalism for the spin connection. Using the complex null frame $\{ {\bf L}, {\bf N}, {\bf M}, \bar{{\bf M}} \}$ the complex spin connection components are labelled as:
\beq \begin{aligned}  & -\kappa = \omega_{311}, -\rho = \omega_{314},~ -\sigma = \omega_{313},~-\tau = \omega_{312}, \\  & \nu = \omega_{422},~ \mu = \omega_{423},~\lambda = \omega_{424},~\pi = \omega_{421}, \\
& -\epsilon = \frac12 (\omega_{211} -  \omega_{431}),~ -\beta = \frac12 (\omega_{213}-\omega_{433}), \\
& \gamma = \frac12(\omega_{122}- \omega_{342}),~ \alpha = \frac12(\omega_{124}-\omega_{344}).
\end{aligned} \label{NPspin} \eeq

\noindent The spin connection will not transform covariantly under a Lorentz frame transformation; however, the boost weights of the components can be determined:
\beq \text{ b.w. 2 } &:& \kappa \nonumber \\
\text{ b.w. 1 } &:& \rho, \sigma, \epsilon, \nonumber \\
\text{ b.w. 0 } &:& \tau, \pi, \alpha, \beta \label{NPspinBW} \\
\text{ b.w. -1 } &:& \mu, \lambda, \gamma \nonumber \\
\text{ b.w. -2 } &:& \nu \nonumber \eeq

\begin{thm} \label{thm:NEF}

If a teleparallel geometry admits a coframe,
$\{ \bn, \bell, \bar{\bm}, \bm\}$, where the torsion tensor is of alignment type {\bf II} and
\beq \kappa=\sigma=\rho = 0, \label{NonExpandingFrame} \eeq
then coordinates can be chosen so that the coframe is of the form:
\beq \begin{aligned} \bell & = du \\
\bn &= dv + H(u,v, x^1, x^2) du + W_i(u,v,x^1,x^2) dx^i \\
\bm &= M_0(u,v,x^1, x^2) du + M_1(u,x^1, x^2) (dx^1 + i dx^2) \end{aligned} \label{GAKundtframe} \eeq
\noindent where $H$ and $W_i$ are real-valued functions and $M_0$ and $M_1$ are complex-valued functions.

\end{thm}

\begin{proof}

If such a coframe exists, then we can apply a boost and spin to set $\epsilon = 0$ and preserve the conditions of the theorem. Then, the first Cartan structure equation implies
\beq d \bell &= -{\bomega}^2_{~b} \wedge {\bf {h}}^b + \frac12 T^2_{~bc} {\bf {h}}^b \wedge {\bf {h}}^c \\
&= {\bomega}_{1b} \wedge {\bf {h}}^b - \frac12 T_{1bc} {\bf {h}}^b \wedge {\bf {h}}^c.  \eeq

Using the requirement on the spin coefficients in equation  \eqref{NonExpandingFrame} and that the torsion tensor must be of alignment type {\bf II}, so that all positive boost-weight components vanish, the exterior derivative of $\bell$ can be expressed as the wedge product:
\beq d \bell =  {\bf W} \wedge \bell \eeq
\noindent where ${\bf W}$ is some one-form with components as sums of the components of the spin connection and torsion tensor, respectively. 

This satisfies all the conditions of the {\it Frobenius theorem}, and thus there is locally two non-constant, non-zero, real functions $F$ and $u'$ such that $\bell = F du'$. Another non-constant, non-zero real function $u$ can be defined so that $du = F du'$. In addition, if $v$ denotes an affine parameter, then a local coordinate system $\{y^a\} = \{ u,v,x^1,x^2\}$ can be introduced where ${\bf L} = \partial_v$. 

Then completing the complex null coframe adapted to $\bell$ gives:
\beq \begin{aligned} \bell & = du \\
\bn &= dv + H(u,v, x^1, x^2)  du + W_i(u,v,x^1,x^2) dx^i \\
\bm &= M_0(u,v,x^1, x^2) du + M_i(u,v, x^1, x^2) dx^i \end{aligned} \nonumber \eeq
\noindent where $H$ and $W_i$ are real-valued functions and $M_0$ and $M_i$ are complex-valued functions. 

By requiring that the positive b.w. terms of the torsion tensor vanish along with $\sigma$ and $\rho$ we find
\beq \begin{aligned} \sigma &= -i \star (\bell \wedge \bm \wedge d \bm) = \frac{ \bar{M}_{2,v} \bar{M}_1 - \bar{M}_{1,v} \bar{M}_2}{-M_1 \bar{M}_2 + M_2 \bar{M}_1}, \\
\rho &= \frac{i}{2} \star (\bell \wedge ( \bar{\bm} \wedge d \bm - \bm \wedge d \bar{\bm} - \bn \wedge d \bell) ) \\ &= - \frac{M_{1,v} \bar{M}_2 + \bar{M}_{2,v} M_1 - \bar{M}_{1,v} M_2 - M_{2,v} \bar{M_1}}{M_1 \bar{M}_2 - M_2 \bar{M}_1}. \end{aligned} \eeq
From these two equations, it follows that $M_{i,v} = 0$ and we can employ a general coordinate transformation of the form:
\beq (u,v,x^i) \to (u,v, \zeta, \bar{\zeta}),~~ \zeta = \frac{1}{\sqrt{2}}[\tilde{x}^1(u,x^i) + i \tilde{x}^2(u,x^i)], \nonumber \eeq
\noindent where  $\tilde{x}^i$ are real-valued functions. This yields the desired form of the spatial part of the coframe in equation \eqref{GAKundtframe}, without changing the the form of the NP quantities.

\end{proof}

The null coframe \eqref{GAKundtframe} can be related to the coframe \eqref{AKundtframe} using a null rotation about $\bell$ and a spin. However, these Lorentz transformations will not set $\tau$ to zero and so the aligned Kundt coframes can only be aligned to a subset of the Kundt coframes (those with vanishing $\tau$). As we will work with an $\bell$-proper frame, we will avoid using the aligned Kundt frame as an anzatz. We note that the corresponding metric arising from the frame, $g_{\mu \nu} = -2 \ell_{(\mu} n_{\nu)} + m_{(\mu} \bar{m}_{\nu)}$ is the metric for the most general non-twisting expansion-free geometry, relative to the Levi-Civita connection. The corresponding frame derivative operators are:

\beq \begin{aligned} D = L^\mu \partial_\mu &=  \partial_v \\
\delta  = M^{\mu} \partial_\mu &= -\left[ \frac{W_1 + i W_2 }{2 \bar{M}_1} \right] \partial_v + \frac{1}{2 \bar{M}_1} ( \partial_{x^1} + i \partial_{x^2}) \\ 
\Delta = N^{\mu} \partial_\mu &=  \partial_u + \left[ \frac{ Im(M_0 \bar{M}_1) W_2 + Re(M_0 \bar{M}_1) W_1}{|M_1|^2} - H \right ] \partial_v \\ 
&  - \left[\frac{Re(M_0 \bar{M}_1)}{|M_1|^2} \right] \partial_{x^1} - \left[\frac{ Im(M_0 \bar{M}_1)}{|M_1|^2} \right] \partial_{x^2}. \end{aligned} \label{KundtDualFrame} \eeq

We have constructed a zeroth order Cartan frame, by aligning $\bell$ with the irreducible parts of the torsion tensor. While there may be some additional Lorentz gauge freedom (i.e., spins, boosts, null rotations about $\bell$) that can be fixed using the covariant derivatives of the torsion tensor,  it will not be necessary to exhaust this freedom. We will call such a coframe \eqref{GAKundtframe} together with a non-trivial spin connection such that $\kappa=\rho = \sigma = \epsilon = 0$ a {\bf Kundt coframe}. We note that relative to a Kundt coframe, the covariant derivatives of the torsion tensor are permitted to be of alignment type {\bf I}. 

By imposing the $\mathcal{I}_T$-degenerate condition and taking covariant derivatives of the irreducible parts of the torsion tensor we can determine strong conditions on the frame and the arbitrary spin connection. We will call any Kundt coframe satisfying the $\mathcal{I}_T$-degenerate condition a {\bf degenerate Kundt coframe}. 

\begin{thm} \label{thm:DKundt}
For any $\mathcal{I}_T$-degenerate teleparallel geometry, coordinates can be chosen so that the Kundt coframe takes the form:
\beq \begin{aligned} \bell & = du \\
\bn &= dv + H(u,v, x^1, x^2)  du + W_i(u,v,x^1,x^2) dx^i \\
\bm &= M_0(u,v,x^1, x^2) du + M_1(u,x^1, x^2) (dx^1 + i dx^2) \end{aligned} \label{DegKundtFrame} \eeq
\noindent where $M_0$ and $M_1$ are complex-valued functions, $H$ and $W_i$ are real-valued functions and $M_0, H$ and $W_i$ are polynomial in the $v$-coordinate:
\beq M_0 = M_0^{(1)}v + M_0^{(0)},~ H =  H^{(2)} v^2 + H^{(1)} v + H^{(0)} \text{ and } W_i = W_i^{(1)} v+ W_i^{(0)}. \label{DegKundtFns} \eeq

\noindent Relative to this coframe, the following spin connection components must vanish
\beq \kappa=\rho=\sigma=\epsilon = 0 \eeq

\noindent and an $\bell$-preserving Lorentz transformation can be used to set all other spin connection components to zero except $\tau = \tau(u,x,y)$.

\end{thm}

\begin{proof}
By imposing the assumption of lemma \ref{lem:AlgSpcl}, we will consider the irreducible parts of the torsion tensor separately and require that their covariant derivatives are of alignment type {\bf II}. To prove the above result we will exploit the fact that for a generic $\mathcal{I}_T$-degenerate teleparallel geometry, we may start in a proper coframe and make a null rotation about $\bn$ to fix $\bell$ which introduces $\kappa,~ \sigma,~\rho$ and $\tau$. This implies that given a Kundt coframe which is improper but $\bell$ is aligned with the torsion tensors, then it is always possible to employ boosts, spins and null rotations about $\bell$ to set all other spin components to zero except $\kappa, \sigma, \rho$ and $\tau$. That is, this is an $\bell$-proper coframe. 

Defining $D = L^\mu \partial_\mu,~\Delta = N^{\mu} \partial_\mu,~~\delta = M^{\mu} \partial_\mu$, we will find conditions on the coframe derivatives of the torsion tensor components. To distinguish between the frame derivative operator $D$ and the covariant derivative $D_a$ without an index, we will use ${\bf D}$. 

To start we will examine the vector part of the torsion tensor. If the vector part of the torsion is of alignment type {\bf O} then the axial part of the torsion can be studied in the same manner. In this analysis we will consider the vector and axial parts of the torsion as one-forms.

\begin{itemize}
\item type {\bf II}: We may choose a frame so that ${\bf V
} = V_3 \bh^3$, then in order for the positive boost weight  components  ${\bf D} {\bf V}$ to vanish, we find
\beq \kappa = \rho = \sigma = 0 \text{ and } D V_3 = (\epsilon - \bar{\epsilon})V_3 \eeq
\noindent As the torsion tensor is of alignment type {\bf II}, this implies we can choose coordinates and an improper coframe as in theorem \ref{thm:NEF}. Making a transformation to an $\bell$-proper coframe, the vector part of torsion (treated as a one-form) transforms as, $V = V_2 \bell + V_3 \bar{\bm} + V_4 \bm $ with $V_4 = \bar{V_3}$ and since $\epsilon = 0$ in this $\bell$-proper frame $DV_3 = DV_4 = 0$. At second order, the vanishing of the positive boost weight components of ${\bf D}^2{\bf V}$ imply that \beq D^2 V_2 = 0,~D \tau V_3 + D \bar{\tau} V_4 = 0. \nonumber \eeq
\noindent We note that $D^2 V_2 = 0$ implies that $D V_2 = f(u,x^i)$ and so by applying a boost and spin to set $V_3 = V_4$ and $D V_2 = 1$, and introduce $\alpha, \beta, \gamma$ and $\epsilon$ as non-zero spin connection components which are independent of $v$ as they arise in the transformation rules for the spin-coefficients as frame derivatives of $D V_2$. Then the vanishing positive boost weight components at third and fourth order yield differential equations for $\tau$ and $V_3$ that require $D \tau = 0$.

\item type {\bf III}: Due to the algebraic type, it follows that ${\bf V} = V_2 \bell$ and we can transform the coframe to a $\bell$-proper coframe without lost of generality. At first order, requiring that ${\bf D} {\bf V}$ is of type {\bf II} implies that $\kappa = 0$. At second order, the spin-coefficients $\rho$ and $\sigma$ must vanish and $D^2 V_2 = 0$. Thus, we may choose coordinates and a coframe as in theorem \ref{thm:NEF}. Continuing to the components of ${\bf D}^3 {\bf V}$, we find
\beq D \tau D V_2 = 0,~ D^2 \tau V_2 =0 \nonumber \eeq
\noindent It is possible that $D V_2 = 0$ and so $D^2 \tau =0$, however, looking at the positive boost weight components of ${\bf D}^4 {\bf V}$, we find $D \tau = 0$.

\end{itemize}

In the case that both the vector and axial parts of the torsion tensor are of alignment type {\bf O}, then we may look to the purely tensorial part, ${\bf t}$.

\begin{itemize}
\item type {\bf II}: At first order, the positive boost weight components of ${\bf D} {\bf t}$ vanish if and only if
\beq \kappa=\tau=\sigma = 0,~ D t_{(13)2} = D t_{(23)1} = 0. \nonumber \eeq
\noindent Thus theorem \ref{thm:NEF} implies we can use coordinates and the coframe \eqref{GAKundtframe}. Choosing the $\bell$-proper coframe and computing ${\bf D}^2 {\bf t}$ the vanishing of the positive boost weight terms yields
\beq D^2 t_{(12)2}= 0, D^2 t_{(33)2} = 0, D^2 t_{(24)3} =0,~ D \tau = 0. \nonumber \eeq
\noindent Finally, in order for $D^3 {\bf t}$ to be of alignment type {\bf II}, there is a differential condition on the boost weight $-2$ terms, $D^3 t_{(23)2} = 0$.

\item type {\bf III}: Imposing that the first covariant derivative, ${\bf D} {\bf t}$, is of type {\bf II}, we find $\kappa=0$. At second order, ${\bf D}^2 {\bf t}$, the vanishing positive boost weight components give
\beq \rho = \sigma = 0, D^2 t_{(12)2}= 0, D^2 t_{(33)2} = 0, D^2 t_{(24)3} =0. \nonumber \eeq
\noindent Choosing the $\bell$-proper coframe, ${\bf t}$ remains of alignment type {\bf III} and  the positive boost weight components of ${\bf D}^3 {\bf t}$ give:
\beq D^2 t_{(12)2}= 0, D^2 t_{(33)2} = 0, D^2 t_{(24)3} =0. \nonumber \eeq
\noindent Finally, computing the fourth covariant derivative, there is one more identity that must be satisfied which implies $D \tau = 0$.

\item type {\bf N}: As the tensor ${\bf t}$ is invariant under null rotations about $\bell$, we will choose a coframe where all spin-components are zero except $\kappa, \rho, \sigma$ and $\tau$. At first order ${\bf D} {\bf t}$ is at most of alignment type {\bf II} and we must consider the second covariant derivative, ${\bf D}^2 {\bf t}$. To simplify matters, we will consider the Bianchi identities \cite{Aldrovandi_Pereira2013}, $T^a_{~[bc|d]} = 0$, when ${\bf V} = {\bf A} = 0$, giving

\beq D t_{(23)2} = (\bar{\rho}+\rho) t_{(23)2}, \kappa = 0, t_{(24)2} \sigma = \rho t_{(23)2}, \tau t_{(24)2} - \bar{\tau} t_{(23)2} = \delta t_{(24)2} - \bar{\delta} t_{(23)2}. \nonumber \eeq

\noindent The vanishing of $\kappa$ is sufficient to ensure ${\bf D}^2 {\bf t}$ is of type {\bf II}. Looking at ${\bf D}^3 {\bf t}$ and setting the positive boost weight components to zero, we find that
\beq \rho=\sigma = 0,~~ D t_{(23)2} = 0. \eeq

\noindent To show that $D \tau = 0$, we will consider the vector-trace of ${\bf D}^2 {\bf t}$, $W_c = t^{(ab)}_{~~~c;ab}$:
\beq {\bf W} = [D \bar{\tau} t_{(23)2} + D \tau t_{(24)2}] \bell \eeq

\noindent If this vector field is non-zero, then it is of type {\bf III} and taking covariant derivatives of ${\bf W}$ will show that $D \tau =0$. If the vector field vanishes, then differentiating the Bianchi identities with respect to $D$ gives:
\beq D \bar{\tau} t_{(23)2} - D \tau t_{(24)2} = 0 \eeq

\noindent Thus $D \tau =0$ otherwise the tensorial part of the torsion is of type {\bf O} and the teleparallel geometry is Minkowski space.

\end{itemize}

To complete the proof, we consider the contribution of the anholonomy coefficients of the frame version of the coframe in equation \eqref{KundtDualFrame}. Noting that $\tau$ is non-zero, with $D \tau = 0$, the anholonomy components must satisfy the same conditions as the components of the torsion tensor. That is, relative to this coordinate system, the components must be polynomial in $v$ with degree equal to the absolute value of their boost weight. Solving the resulting differential equations gives \eqref{DegKundtFns}.

\end{proof}

The coframe in \eqref{DegKundtFrame} is considerably more complicated than the initial Kundt coframe in \eqref{AKundtframe}. This is due to the fact that the coframe in \eqref{AKundtframe} is not $\bell$-proper and will have a more general spin connection.

Without loss of generality, the frame in equation \eqref{DegKundtFrame} can be considered as a $\bell$-proper coframe. We note that this frame is not necessarily a Cartan frame as it is not invariantly defined by normalizing the components of the torsion tensors.

\begin{lem}
For any $\mathcal{I}_T$-degenerate teleparallel geometry, the Kundt coframe in equations \eqref{DegKundtFrame} and \eqref{DegKundtFns} is related to a proper coframe by a null rotation about $\bn$ with the frame functions, $W_i$ and $M_0$ and null rotation parameter satisfy the following constraints:
\beq \begin{aligned} E &= E^{(1)} (u,x^1, x^2) v +E^{(0)} (u,x^1, x^2). \\  0 &= \Delta E,\\  0 &= \bar{\delta} E, \\
 0 &=  \delta E. \end{aligned}  \eeq
\end{lem}

\begin{proof}

To characterize the entire set of $\mathcal{I}_T$-degenerate teleparallel geometries, we will momentarily consider the transformation from the ${\mbold \ell}$-proper coframe to a proper coframe $\{ {\bf n'}, {\mbold\ell'},  \bar{{\bf m'}}, {\bf m'} \}$, in order to align the torsion tensor and its covariant derivatives a null rotation about $\bn'$ is necessary:
\beq {\bf n'} = \bn,~~ {\mbold\ell'} = \bell + \bar{E} \bm + E \bar{\bm} + |E|^2 \bn,~~ {\bf m'} = \bm + E \bn. \eeq

Doing so, the only non-zero spin connections are:
\beq \begin{aligned} \kappa' &= |E|^2 \tau + \Delta E + E \delta \bar{E} + \bar{E} \bar{\delta} E + |E|^2 D E,\\ \sigma' &= E \tau + E D E + \bar{\delta} E, \\ 
\rho' &= \bar{E} \tau + \bar{E} D E + \delta E,\\
 \tau' &= \tau + D E. \end{aligned} \eeq

Requiring that these quantities vanish we find the simpler equations 
\beq \begin{aligned} \tau &=- D E, \\
 0 &= \Delta E ,\\
 0 &= \bar{\delta} E, \\
 0 &= \delta E,~. \end{aligned} \label{DegKundtReduced}\eeq

\noindent Here, $E$ is linear in $v$ in order to satisfy theorem \ref{thm:DKundt}. Equivalently, for each choice of $E$, the equations in \eqref{DegKundtReduced} are linear equations for the frame functions $H, W_i$ and $M_0$ once expanded in terms of equations \eqref{KundtDualFrame}.

\end{proof}

\begin{thm} \label{thm:ImproperDKundt}
For any $\mathcal{I}_T$-degenerate teleparallel geometry with an $\bell$-proper Kundt coframe of the form \eqref{DegKundtFrame} and \eqref{DegKundtFns} with $\tau \neq 0$, the null rotation parameter $E(u,x^i)$ is real-valued and the frame functions take the following form: \beq \begin{aligned} W_i &= ( \ln E^{(1)})_{,x^i}v + \frac{E^{(0)}_{,x^i}}{E^{(1)}}  \\ 
H^{(2)} &=0 \\ 
H^{(1)} &= ( \ln E^{(1)})_{,u} \\
H^{(0)} &= \frac{E^{(0)}_{,u}}{E^{(1)}}.  \label{DegKundtfnsReduced} \end{aligned} \eeq

\end{thm}

\begin{proof}
Expanding the conditions \eqref{DegKundtReduced} using the frame basis \eqref{KundtDualFrame} and the conditions on the frame functions in \eqref{DegKundtFns}, the conditions on the frame functions can be solved algebraically from the coefficients of the powers of the $v$-coordinate.
\end{proof}

We note that the class of $\mathcal{I}_T$ teleparallel geometries where the coframe in theorem \ref{thm:DKundt} is proper constitutes a significantly larger class of coframes the aligned Kundt coframes.

\subsection{VSI$_T$ teleparallel geometries}
With this approach, we can determine the class of VSI$_T$ teleparallel geometries \cite{Hervik2011}:

\begin{cor}
The teleparallel geometries where the torsion tensor and all of its covariant derivatives are of type {\bf III} or more special, constitute the class of VSI$_T$ teleparallel geometries.
\end{cor}

\noindent Using the $\bell$-proper coframe and the above corollary we can prove the following result for all VSI$_T$ teleparallel geometries.

\begin{thm}
The class of VSI$_T$ teleparallel geometries are given by the coframe:

\beq \begin{aligned} \bell & = du \\
\bn &= dv + H(u,v, x^1, x^2)  du + W_i(u,x^1,x^2) dx^i \\
\bm &= M_0(u,x^1, x^2) du + M_1(u) (dx^1 + i dx^2) \end{aligned} \label{VSIFrame} \eeq
\noindent where $M_0$ and $M_i$ are complex-valued functions, $H$ and $W_i$ are real-valued functions and  $H$ is linear in the $v$-coordinate:
\beq H = H^{(1)}(u,x^1,x^2) v + H^{(0)}(u,x^1,x^2). \label{VSIFns} \eeq

\noindent This coframe is proper, as the spin connection must be zero in this frame.

\end{thm}

\begin{proof}

Taking an $\bell$-proper frame of the form \eqref{DegKundtFrame} with frame functions in \eqref{DegKundtFns}, then computing the torsion tensor and requiring that it is of alignment type {\bf III}, we find the following equations:

\beq \begin{aligned}
& \tau = 0  \\
& D W_1 + i D W_2 = 0 \\
& D M_0 = 0 \\
& \delta M_{1} + i \bar{\delta} M_{1} = 0. \end{aligned} \eeq

\noindent Solving the differential equation for $M_1=M^R_1+iM^{I}_1$, it follows that $M_1$ must be a complex valued function of the $u$-coordinate alone. At first order, ${\bf D} {\bf T}$ is of type {\bf III} if and only if $H^{(2)} = 0$.
\end{proof}

To relate this result to the VSI$_T$ example given in section \ref{sec:Examples}, we can apply a Lorentz transformation which leaves $\bell$ invariant to the coframe in equation \eqref{VSIFrame} with the conditions \eqref{VSIFns} and produce the aligned Kundt coframe \eqref{alignedKundtCoframe} with $\epsilon =0$ along with a spin connection with non-zero components. Thus, any VSI$_T$ teleparallel geometry can be constructed using an aligned Kundt coframe arising from the class of VSI$_{\tilde{R}}$ metrics with $\epsilon =0$, along with a spin-connection arising from any Lorentz transformation that leaves $\bell$ unchanged. This implies that the class of VSI$_T$ teleparallel geometries arises as a subclass of the Kundt frames arising from the VSI$_{\tilde{R}}$ class of metrics.

The explicit relationship between VSI$_T$ and VSI$_{\tilde{R}}$ geometries (that is, how the curvature tensor and its covariant derivatives are expressed in terms of the torsion tensor and its covariant derivatives) is not fully understood. However, we have shown that the choice of possible proper frames arising from the class of VSI$_{\tilde{R}}$ metrics which produce VSI$_T$ teleparallel geometries are restricted to the subclass of $\epsilon = 0$ metrics.

\section{Discussion} \label{sec:discussion}

In this article we have shown that the class of four-dimensional $\mathcal{I}_T$-degenerate teleparallel geometries are described either by a proper Kundt coframe in \eqref{DegKundtFrame} and \eqref{DegKundtFns} or as an improper frame  in \eqref{DegKundtFrame} and \eqref{DegKundtfnsReduced} where the spin connection arises from a null rotation about the null frame element $\bn$. As a simple application of these results we have also provided an explicit form for the VSI$_T$ teleparallel geometries as a choice of proper frame in \eqref{VSIFrame} and \eqref{VSIFns}. Using this proper frame, we have then shown that the class of VSI$_T$ teleparallel geometries can be constructed using the coframe \eqref{alignedKundtCoframe} with $\epsilon =0$ along with a spin-connection arising from a Lorentz transformation that preserves $\bell$.

As in the case of 4D Riemannian geometries with Lorentzian signature, the entire class of $\mathcal{I}_T$-degenerate teleparallel geometries are contained within the Kundt class of metrics. However, there appears to be significant differences, such as the existence of teleparallel geometries which yield the same metric but are different teleparallel geometries. For example, take a proper Kundt coframe of the form \eqref{DegKundtFrame} and \eqref{DegKundtfnsReduced} and an improper frame of the form \eqref{DegKundtFns} and \eqref{DegKundtfnsReduced} and a spin connection arising from a null rotation about $\bn$ with real-valued parameter $E(u,x^i)$. The resulting teleparallel geometries will, in general, have differing canonical forms for their respective torsion tensors. Due to this, the Cartan-Karlhede algorithm will produce entirely distinct sets of Cartan invariants for each teleparallel geometry and this is sufficient to prove that they are not related by a coordinate transformation. 

It is possible to distinguish teleparallel gravity theories using observations from cosmology \cite{Nunes:2016plz, Nunes:2016qyp,Capozziello:2015rda} or by studying the propagation of gravitational waves \cite{Hohmann:2018jso}. In addition, these observations can partially discern the difference between solutions with the same metric but differing teleparallel geometries by measuring the components of the torsion tensor. In principle, torsion invariants, like TSPIs, and the Cartan invariants represent quantities that can be measured locally and hence can fully distinguish between two given solutions. 

We have shown that not all of the metrics in the VSI$_{\tilde{R}}$ class can produce a coframe which yields a VSI$_T$ teleparallel geometry, regardless of the choice of spin-connection.  This indicates there is a significant difference between $\mathcal{I}_T$ and the set $\mathcal{I}_{\tilde{R}}$, as a given teleparallel geometry can have a non-empty set $\mathcal{I}_T$ but an empty set $\mathcal{I}_{\tilde{R}}$ when considered as a Lorentzian geometry. This difference arises from the fact that the curvature tensor can be expressed in terms of the torsion tensor and its first covariant derivative with respect to the spin-connection. In future work we will examine the class of teleparallel geometries where all TSPIs are constant, the so-called {\it CSI teleparallel geometries}, and determine if $\mathcal{I}_T$ and $\mathcal{I}_{\tilde{R}}$ encode the same information. 

We believe that this result can be extended to higher dimensions to effectively prove that all $n$-dimensional $\mathcal{I}_T$-degenerate teleparallel geometries are described by a Kundt coframe and a choice of spin connection arising from a null rotation about $\bn$ in order to account for the possibility that $\bell$ must be aligned with the torsion tensors. Due to the form of the irreducible parts of the tensor, the proof in higher dimensions would be straightforward when the vector part of the torsion tensor is non-zero. The primary obstacle lies in the decomposition of the trace-free part of the torsion tensor and the application of the alignment classification to this tensor.  If this  is true, it may provide a stronger argument for the conjecture that the class of $\mathcal{I}_{\tilde{R}}$-degenerate spacetimes lie entirely in the Kundt class of metrics, known as the Kundt conjecture \cite{CHPP2009}. In 4D, there are examples of TSPIs which are topological, such as the torsional equivalent of the Gauss-Bonnet term \cite{Kofinas:2014owa}. In higher dimensions, such topological invariants may no longer be topological and provide additional information on the teleparallel geometries such as their role in syzygies of other TSPIs.

\section*{Acknowledgments}
AAC was supported by the Natural Sciences and Engineering Research Council of Canada. 


\bibliographystyle{apsrev4-2}
\bibliography{Tele-Parallel-Reference-file1}

\end{document}